\documentclass[11pt]{article}

\usepackage{amssymb}
\usepackage{amsmath}
\usepackage{amsthm}
\usepackage{fancyhdr}
\usepackage{graphics}
\usepackage{graphicx}
\usepackage{environ}
\usepackage{framed}
\usepackage{url}
\usepackage[margin=1.0in]{geometry}

\newcommand{\PSPACE}{\mathrm{PSPACE}}

\newcommand{\BPE}{\mathrm{BPE}}

\theoremstyle{plain}
\newtheorem{lemma}{Lemma}[section]
\newtheorem*{lemma*}{Lemma}
\newtheorem{theorem}[lemma]{Theorem}

\newtheorem*{corollary*}{Corollary}

\theoremstyle{definition}
\newtheorem{definition}[lemma]{Definition}

\begin{document}
\title{A Note on Unconditional Subexponential-Time Pseudo-deterministic Algorithms for BPP Search Problems}
\author{Dhiraj Holden}
\maketitle
\begin{abstract}
We show the first unconditional pseudo-determinism result for all of search-BPP. Specifically, we show that every BPP search problem can be computed pseudo-deterministically on average for infinitely many input lengths. In other words, for infinitely many input lengths and for any polynomial-time samplable distribution our algorithm succeeds in producing a unique answer (if one exists) with high probability over the distribution and the coins tossed. 
\end{abstract}

\section{Introduction}
Gat and Goldwasser \cite{GG11} introduced the notion of pseudo-deterministic algorithms for search problems; a pseudo-deterministic algorithm $A(x,r)$ has the property that $\Pr_{r_1,r_2}[A(x,r_1) = A(x,r_2) ] \geq 1 - 1/poly(n)$, which means that the algorithm finds the same output with high probability over the randomness. The concept of pseudo-deterministic algorithms has been well-studied, giving pseudo-deterministic algorithms that improve on deterministic algorithms for a wide variety of problems such as finding non-zero evaluations of multi-variate polynomials, finding primitive roots of primes, and finding bipartite matchings in parallel \cite{G15,GGR,matching,GG11,GGH}.

Recently, Oliveira and Santhanam \cite{OS} show that for the general question ``given  length $n$, find a string of length $n$ with a pre-specified  {\it dense property}'', pseudo-deterministic algorithms exist which
for infinitely many lengths $n$  yield  a unique string with said property of length $n$ with high probability. 
An example of such a dense property is the set of prime numbers. 

The major open question is this: does every BPP search problem have a polynomial-time pseudodeterministic algorithm? 
In this work we will show that every BPP search problem has an algorithm for infinitely many input lengths running in subexponential time that for every polynomial-time samplable distribution produces a unique answer with high probability on inputs drawn from the distribution and over the random coins.

This work expands on the work of Oliveira and Santhanam. \cite{OS} give a pseudo-deterministic algorithm for estimating the acceptance probability of a circuit; we extend their work to general search-BPP problems, where the input is a string over some alphabet and our algorithm outputs a string that satisfies a relation with the input. We combine this with the algorithm to generate canonical samples from a polynomial-time samplable distribution given in \cite{OS} to achieve an average-case pseudo-deterministic algorithm for every problem in search-BPP. 

\subsection{Outline of the proof}
The proof of our main theorem uses techniques from \cite{OS} and extends them to the case of general search-BPP problems, in particular the results of \cite{IW01,TV07} on pseudorandomness assuming uniform assumptions. There are two cases to the proof; either PSPACE is contained in subexponential-time BPP, or it is not. In the first case we can use this to get subexponential-size circuit lower bounds for BPE, which extending an argument of \cite{OS} implies a subexponential-time pseudo-deterministic algorithm for every search-BPP problem on infinitely many input lengths. If PSPACE is not contained in subexponential-time BPP, \cite{TV07} gives a construction of a pseudorandom generator which is not distinguishable uniformly for infinitely many input lengths. This implies that in particular, the search-BPP algorithm cannot distinguish the output of the PRG from uniform when fed inputs from any polynomial-time samplable distribution, which we use to construct an algorithm running in subexponential time that for infinitely many input lengths outputs a unique answer with high probability over the input distribution and the randomness of the algorithm when the input is drawn from any polynomial-time samplable distribution. 

\section{Preliminaries}
In this section we will define the notions used in this paper. We will define pseudo-determinism for BPP search problems and the notion of average-case complexity we will use. First, we will define a search problem.

\begin{definition}[Search Problem]
A search problem is a relation $R$ consisting of pairs $(x,y)$. We say that an algorithm $A$ solves a search problem $R$ if $(x,A(x)) \in R$ for all $x$. 
\end{definition}

With this definition, we can now given the definition of search-BPP.

\begin{definition}[Search-BPP]
A binary relation $R$ is in search-BPP if there exists an algorithm $A$ running in probabilistic polynomial time that for every $x$ outputs a $y$ such that $(x,y) \in R$ with probability at least 2/3, and there exists a BPP machine $B$ such that if it accepts on input $(x,y)$, then $(x,y) \in R$ and also $B$ accepts $(x,A(x,r))$ with probability at least 1/2. 
\end{definition}

A pseudo-deterministic search-BPP problem is a search-BPP problem with a \\ pseudo-deterministic algorithm, or an algorithm which outputs the same $y$ with high probability.

\begin{definition}[Pseudo-deterministic search-BPP]
A search-BPP relation $R$ is in \\ pseudo-deterministic search-BPP if there exists an algorithm $A$ such that for every $x$ there exists a $y$ such that $(x,y) \in R$ and $\Pr[A(x,r) = y] \geq 1/2$. 
\end{definition}

To obtain the best running time for our pseudo-deterministic algorithm, we will need the iterated exponential functions first used in complexity theory by \cite{MVW99}. We will be considering functions with half-exponential growth, i.e. functions $f$ such that $f(f(n)) \in O(2^{n^k})$ for some $k$.

\begin{definition}[Fractional exponentials \cite{MVW99}]
The fractional exponential function $e_{\alpha}(x)$ will be defined as $A^{-1}(A(x) + \alpha)$, where $A$ is the solution to the functional equation $A(e^x - 1) = A(x) + 1$. In addition, we can construct such functions so that $e_{\alpha}(e_{\beta}(x)) = e_{\alpha + \beta}(x)$. It is clear from this definition that $e_{1}(n) = O(2^n)$ as desired.  
\end{definition}

In addition, we also need to talk about average-case complexity for search problems. Average-case complexity is defined over a given distribution, though our results will extend to every polynomial-time samplable distribution.

\begin{definition}[Average-case search-BPP]
We say that a search problem given by a relation $R$ is in HeurBPTIME($t(n)$,$\delta(n)$) for a given distribution $\mathcal{D}$ if there is some algorithm $A$ running in time $t(n)$ such that $\Pr_{x \leftarrow D_{|x|}, r}[ (x,A(x,r)) \in R ] \geq 1 - \delta(n)$, and there exists a BPP machine $B$ such that if it accepts on input $(x,y)$, then $(x,y) \in R$ and also $B$ accepts $(x,A(x,r))$ with probability at least 1/2.
\end{definition}

We can define average-case pseudo-deterministic search-BPP in a similar fashion by requiring that $A(x,r)$ be unique with high probability.

\begin{definition}[Average-case pseudo-deterministic search-BPP]
We say that a search problem given by a relation $R$ is in HeurPsdTIME($t(n)$,$\delta(n)$) for a given distribution $\mathcal{D}$ if there is some algorithm $A$ running in time $t(n)$ such that  for every $x$ there exists a $y(x)$ such that $(x,y) \in R$ and $\Pr_{x \leftarrow D_{|x|}, r}[  A(x,r) = y(x) ] \geq 1 - \delta(n)$.
\end{definition}

The final result that we need is a statement of \cite{TV07} about a pseudorandom generator based on hardness for PSPACE. 

\begin{theorem}[Corollary 4.4 of \cite{TV07}]
\label{prg}
For every function $f \in \PSPACE$, there is a constant $d$ such that if $f \notin \cup_{c} \mathrm{BPTIME}(t(n^d)^c)$, then there is a generator $G$ with stretch $t(\cdot)$ that cannot be $1/t(\cdot)^c$-distinguished uniformly in time $t(\cdot)^c$ for any constant $c$. 
\end{theorem}

\section{Results}
In this section we will show that every relation $R$ in search-BPP is contained in \\ HeurPsdTIME($e_{1/2 + \epsilon}(n)$,$\frac{1}{poly(n)}$) for infinitely many input lengths for every polynomial-time samplable distribution. 

\begin{theorem}
Let $R$ be a relation in search-BPP. Then for every polynomial-time samplable distribution $\mathcal{D}$, and for every $\epsilon > 0$, $R$ is contained in HeurPsdTIME($e_{1/2 + \epsilon}(n)$,$\frac{1}{poly(n)}$) for infinitely many input lengths.
\end{theorem}

\begin{proof}
There are two cases to the proof. Suppose that $\PSPACE \subseteq \mathrm{BPTIME}(e_{1/2 + \epsilon}(n))$. Then, padding implies that $\mathrm{SPACE}(e_{1/2 - \epsilon}(n)) \subseteq \BPE$, which in turn implies $\BPE$ cannot be approximated by $e_{1/2 - \epsilon}(n)$-size circuits. \cite{OS} notes that this fact implies that any search-BPP relation can be solved pseudo-deterministically in time $O(e_{1/2 + \epsilon}(n))$ for infinitely many input lengths. Now, suppose that $\PSPACE \nsubseteq \mathrm{BPTIME}(e_{1/2 + \epsilon}(n))$. Then, by Theorem \ref{prg}, this means that there is a generator $G$ with stretch $e_{1/2 + \epsilon}(\cdot)$ that cannot be $1/e_{1/2 + \epsilon}(\cdot)^c$-distinguished in time $e_{1/2 + \epsilon}(\cdot)$. Then, we claim that the following algorithm is a HeurPsdTIME($e_{1/2 + \epsilon}(n)$,$\frac{1}{poly(n)}$) algorithm for infinitely many input lengths. Let us call the search algorithm of $R$ $A$ and the verification algorithm of $R$ $B$. We will use the version of $B$ amplified to a $1 - exp(n)$ success probability. Iterate through the $e_{1/2 - \epsilon}(\cdot)$ outputs of $G$ as the randomness, and output the first $A(x,r)$ such that $B$ accepts $(x,A(x,r))$. Suppose that this algorithm is not a HeurPsdTIME($e_{1/2 + \epsilon}(n)$,$\frac{1}{poly(n)}$) algorithm for infinitely many input lengths. Then, this means that this algorithm fails for large enough input lengths, which means that the output of this algorithm is not correct or not unique with high probability over the distribution and the randomness. Since $B$ accepts all $(x,y) \in R$ with high probability, this means that with high probability the first $A(x,r)$ such that $B$ accepts $(x,A(x,r))$ is also the first $A(x,r)$ such that $(x,A(x,r)) \in R$. Now, it suffices to show that there exists such an $A(x,r)$. If no such $A(x,r)$ exists, then $\Pr_{r \leftarrow G}[ B((x,A(x,r)) = 1] \leq exp(-n)$, and $\Pr_{r \leftarrow U}[ B(x,A(x,r)) = 1 ]\geq c$ for some constant $c$ by the correctness guarantees of $A$ and $B$. Thus, if no such $A(x,r)$ exists, $B(x,A(x,r))$ with $x$ drawn from some polynomial-time samplable distribution is a uniform algorithm which distinguishes between the output of $G$ and the uniform distribution, contradicting the fact that $G$ is indistinguishable. Thus the algorithm is a HeurPsdTIME($e_{1/2 + \epsilon}(n)$,$\frac{1}{poly(n)}$) for infinitely many input lengths. 
\end{proof}

\section{Open Problems}

It is open whether this result can be improved further by showing a faster algorithm, making the algorithm work more than just infinitely often, or by giving a worst-case rather than an average-case algorithm. We believe that this would require new techniques in the theory of pseudorandomness. In addition, it is open whether pseudo-deterministic simulation of general search-BPP problems implies circuit lower bounds, or if a theorem similar to that of \cite{KI04} showing that derandomization implies either circuit lower bounds for NEXP or arithmetic circuit lower bounds for computing the permanent can be proven. If so, this will require substantially different techniques than \cite{KI04}, as finding a non-zero value of a polynomial is known to be in pseudo-deterministic polynomial time. 

\section*{Acknowledgements}

I would like to thank Shafi Goldwasser, Ofer Grossman, and Rahul Santhanam for valuable comments on this work. This work was supported by NSF MACS - CNS-1413920 and the
SIMONS Investigator award Agreement Dated 6-5-12.

\bibliographystyle{plain}
\bibliography{bibfile}

\begin{thebibliography}{10}

\bibitem{GG11}
Eran Gat and Shafi Goldwasser.
\newblock Probabilistic search algorithms with unique answers and their
  cryptographic applications.
\newblock In {\em Electronic Colloquium on Computational Complexity (ECCC)},
  volume~18, page 136, 2011.

\bibitem{GGR}
Oded Goldreich, Shafi Goldwasser, and Dana Ron.
\newblock On the possibilities and limitations of pseudodeterministic
  algorithms.
\newblock In {\em Proceedings of the 4th conference on Innovations in
  Theoretical Computer Science}, pages 127--138. ACM, 2013.

\bibitem{matching}
Shafi Goldwasser and Ofer Grossman.
\newblock Perfect bipartite matching in pseudo-deterministic rnc.
\newblock In {\em Electronic Colloquium on Computational Complexity (ECCC)},
  volume~22, page 208, 2015.

\bibitem{GGH}
Shafi Goldwasser, Ofer Grossman, and Dhiraj Holden.
\newblock Pseudo-deterministic proofs.

\bibitem{G15}
Ofer Grossman.
\newblock Finding primitive roots pseudo-deterministically.
\newblock In {\em Electronic Colloquium on Computational Complexity (ECCC)},
  volume~22, page 207, 2015.

\bibitem{IW01}
Russell Impagliazzo and Avi Wigderson.
\newblock Randomness vs. time: De-randomization under a uniform assumption.
\newblock In {\em Foundations of Computer Science, 1998. Proceedings. 39th
  Annual Symposium on}, pages 734--743. IEEE, 1998.

\bibitem{KI04}
Valentine Kabanets and Russell Impagliazzo.
\newblock Derandomizing polynomial identity tests means proving circuit lower
  bounds.
\newblock {\em Computational Complexity}, 13(1-2):1--46, 2004.

\bibitem{MVW99}
Peter~Bro Miltersen, N~Variyam Vinodchandran, and Osamu Watanabe.
\newblock Super-polynomial versus half-exponential circuit size in the
  exponential hierarchy.
\newblock In {\em International Computing and Combinatorics Conference}, pages
  210--220. Springer, 1999.

\bibitem{OS}
Igor~C Oliveira and Rahul Santhanam.
\newblock Pseudodeterministic constructions in subexponential time.
\newblock {\em arXiv preprint arXiv:1612.01817}, 2016.

\bibitem{TV07}
Luca Trevisan and Salil Vadhan.
\newblock Pseudorandomness and average-case complexity via uniform reductions.
\newblock {\em Computational Complexity}, 16(4):331--364, 2007.

\end{thebibliography}
\end{document}